\documentclass[letterpaper]{article} 
\usepackage{aaai2026}  
\usepackage{times}  
\usepackage{helvet}  
\usepackage{courier}  
\usepackage[hyphens]{url}  
\usepackage{graphicx} 
\usepackage{natbib}  
\usepackage{caption} 
\usepackage{float} 
\usepackage{microtype}
\usepackage{amsmath} 
\usepackage{makecell} 
\usepackage{amsthm} 
\usepackage{amssymb} 
\usepackage{mdframed} 
\usepackage{multirow} 
\usepackage{booktabs} 
\usepackage{xcolor} 
\usepackage{amsthm} 

\usepackage[ruled,linesnumbered,vlined]{algorithm2e}
\usepackage{subfig} 
\usepackage{enumitem} 
\usepackage{tabularx}

\usepackage{newfloat}
\usepackage{listings}
\DeclareCaptionStyle{ruled}{labelfont=normalfont,labelsep=colon,strut=off} 
\floatstyle{ruled}
\newfloat{listing}{tb}{lst}{}
\floatname{listing}{Listing}
\newtheorem{definition}{Definition} 

\newtheorem{theorem}{Theorem}
 
\newtheorem{lemma}{Lemma}
\newtheorem{corollary}{Corollary}

\newcommand{\descr}[1]{\noindent \textbf{#1}}

\SetKwFor{Upon}{upon}{do}{}

\title{Improved Streaming Algorithm for Fair $k$-Center Clustering}
\author{
    Longkun Guo\textsuperscript{{\rm 1,2}}\thanks{Corresponding Author},
    Zeyu Lin\textsuperscript{\rm 1},
    Chaoqi Jia\textsuperscript{\rm 3},
    Chao Chen\textsuperscript{\rm 3}
}
\affiliations{
     \textsuperscript{\rm 1} School of Mathematics and Statistics, Fuzhou University, Fuzhou 350116, China\\
     \textsuperscript{\rm 2} School of Integrated Circuits, Guangdong University of Technology, Guangzhou 510006, China\\
    \textsuperscript{\rm 3} School of Accounting, Information Systems and Supply Chain, RMIT University, Melbourne, VIC 3000, Australia\\
    
    longkun.guo@gmail.com,
    zeyu\_lin@foxmail.com,
    \{chaoqi.jia,chao.chen\}@rmit.edu.au
}

\begin{document}
\maketitle

\begin{abstract}
Many real-world applications call for incorporating fairness constraints into the $k$-center clustering problem, where the dataset is partitioned into $m$ demographic groups, each with a specified upper bound on the number of centers to ensure fairness. Focusing on big data scenarios, this paper addresses the problem in a streaming setting, where data points arrive sequentially in a continuous stream. Leveraging a structure called the $\lambda$-independent center set, we propose a one-pass streaming algorithm that first computes a reserved set of points during the streaming process.  In the post-streaming process, we then select centers from the reserved point set by analyzing three possible cases and transforming the most complex one into a specially constrained vertex-cover problem on an auxiliary graph. Our algorithm achieves an approximation ratio of $5+\epsilon$ and memory complexity $O(k\log \alpha)$, where $\alpha$ is the aspect ratio and $\epsilon>0$ is any small constant. Furthermore, we extend our approach to semi-structured data streams, where data points arrive in groups. In this setting, we present a $(3+\epsilon)$-approximation algorithm for $m = 2$, which can be readily adapted to solve the offline fair $k$-center problem, achieving an approximation ratio of 3 that matches the current state of the art. Lastly, we conduct extensive experiments to evaluate the performance of our approaches, demonstrating that they outperform existing baselines in both clustering cost and runtime efficiency. 
\end{abstract}

\begin{links}
    \link{Code}{https://github.com/ChaoqiJia/StreamFKC}
\end{links}

\section{Introduction}
Fair $k$-center clustering is a popular problem in various fields, including data summarization~\cite{kleindessner2019fair,angelidakis2022fair} and machine learning~\cite{chierichetti2017fair,jones2020fair}. The problem has been widely studied, with many definitions of fairness proposed and corresponding approximation algorithms discussed, such as group fairness~\cite{chierichetti2017fair,wu2024new,backurs2019scalable}, data summarization fairness~\cite{kleindessner2019fair,ceccarello2024fast,wu2024new,lin2024streaming}, colorful fairness~\cite{bandyapadhyay2019constant,anegg2022techniques,jia2022fair} and so on. In this paper, we focus on the fair $k$-center problem for data summarization. We consider a dataset $S$ of size $n$ divided into $m$ disjoint groups, denoted as $S = S_{1} \cup \ldots \cup S_{m}$. Our goal is to select $k$ centers to minimize the maximum distance from each point to its nearest center, while the number of centers chosen from each group $S_l$ is bounded by $k_l$. This formulation helps avoid bias toward some sensitive features by controlling the number of objects from each category in the output. For example, it can dictate the number of movies of each genre shown to a user in a recommendation system or limit the number of old messages included in a summary of a user's feed~\cite{mahabadi2024core}. 

A traditional challenge in clustering is the need to process large-scale datasets across numerous applications. In such scenarios, storing the entire input in memory becomes impractical, giving rise to the streaming model. In this model, data points arrive sequentially in a stream, and only a limited portion can be retained in memory due to space constraints. The streaming algorithm must decide, upon the arrival of each data point, whether to store it or discard it. Notably, designing effective streaming algorithms is more challenging than developing offline algorithms, as decisions must be made based on partial information rather than having access to the entire dataset. Motivated by this challenge,  we aim to propose approximation algorithms for the fair $k$-center problem under the streaming setting. 

\subsection{Related Work} 
Fairness is an important concept for many research fields such as data summarization~\cite{kleindessner2019fair}, private security~\cite{0012YP24}, and machine learning~\cite{ZhangWYPTP24,0012WYY0Y25}. In this paper, we briefly review the literature on the $k$-center problem, focusing on studies about data summarization fairness and streaming model settings.

\paragraph{Data summarization fairness $k$-center.} 

The fair $k$-center clustering was formally addressed in the context of data summarization by~\citet{kleindessner2019fair}. They proposed a $5$-approximation algorithm for the case of two groups and a $(3 \times 2^{m-1} - 1)$-approximation for the general case with $m$ groups. Subsequently, \citet{jones2020fair} reduced the approximation ratio to $3$ for arbitrary $m$ groups. Before these two relevant works, the matroid center problem, a generalization of the $k$-center problem that enforces a matroid constraint on the center set rather than a simple cardinality constraint, was studied by~\citet{chen2016matroid}. They developed a $3$-approximation algorithm, although with considerably higher time complexity. More recently, \citet{chen2024approximation} introduced the fair $k$-supplier problem, which involves selecting $k$ facilities from a dataset partitioned into $m$ disjoint groups, subject to group-wise upper bounds on the number of facilities selected. They presented a $5$-approximation algorithm for this problem by maximum matching techniques.

\paragraph{Computation of $k$-center for Big Data.} For large-scale datasets in real-time streaming scenarios, streaming $k$-center clustering was studied by~\citet{matthew2008streaming}, who provided a streaming $(4+\epsilon)$-approximation algorithm with outliers, where up to $z$ input points can be dropped and $O(kz/\epsilon)$ memory is used. Then, a deterministic one-pass streaming algorithm was developed for the same problem by~\citet{ceccarello2019solving}, achieving an improved approximation ratio $(3+\epsilon)$ with $O\big((k+z)(96/\epsilon)^{d}\big)$ memory in $d$ dimensions. For the matroid center problem, a $(17+\epsilon)$-approximation one-pass algorithm with a running time $O_{\epsilon}\big((nk+k^{3.5})+k^{2}\log(\Lambda)\big)$ has been developed by \citet{kale2019small}, where $k$ is the rank of the matroid, $\Lambda$ is the aspect ratio of the metric, and $\epsilon$ terms are hidden by the $O_{\epsilon}$ notation. For the fair $k$-center problem, \citet{chiplunkar2020solve} provided a distributed algorithm achieving an approximation ratio of $(17+\epsilon)$ and with $O(kn/l+mk^{2}l)$ running time for $l$ processors. They also developed a two-pass streaming algorithm with an approximation ratio $3$ for the fair $k$-center problem. Recently, the fair range $k$-center was addressed by~\citet{nguyen2022fair} with an approximation ratio of $13$ and \citet{ceccarello2024fast} employed a coreset-based method for the colorful fair $k$-center problem in doubling metrics, achieving a $(3+\epsilon)$-approximation ratio for sequential, streaming and MapReduce/MPC algorithms. In addition, a one-pass streaming algorithm for representative fair $k$-center clustering by \citet{lin2024streaming} achieved the state-of-the-art approximation ratio of $(7+\epsilon)$. This paper improves the approximation ratio to $(5+\epsilon)$ under a streaming setting by utilizing $\lambda$-independent center sets and modifying the post-streaming algorithms.

\subsection{Our Contributions}
The main contributions can be summarized as follows:
\begin{itemize}
   \item Devise a one-pass streaming algorithm for fair $k$-center, achieving an improved approximation ratio of $5+\epsilon$ while consuming $O(k\log \alpha)$ memory, where $\alpha$ is the aspect ratio and $\epsilon>0$ is any small constant. This significantly improves the previous SOTA ratio $7+\epsilon$ due to \citet{lin2024streaming}. 

   \item For semi-structured data streams where data points of each group $S_l$ are streamed as a batch, we propose a streaming algorithm achieving an approximation ratio of $3+\epsilon$  for $m=2$.

  \item   Conduct extensive experiments on real-world datasets to demonstrate the practical performance gains of our algorithms in clustering accuracy and runtime efficiency. 
\end{itemize}

Notably, we also provide an example demonstrating that our ratio analysis is tight for the algorithm. Moreover, the streaming algorithm with ratio $3+\epsilon$ can be adapted to solve the offline fair $k$-center, achieving a ratio of $3$, which matches the state-of-the-art ratio 3 for the offline setting due to~\citet{jones2020fair}.

\section{Preliminary and Algorithmic Framework}
\label{sec:prob_statement}
Let $S$ be a finite set of data points with size $n$ distributed in a metric space, where $d: S\times S\rightarrow\mathbb{R}_{\geq 0}$ is the distance on $S$ that satisfies \textit{the triangle inequality}. For a given parameter $k\in\mathbb{N}$, the traditional $k$-center clustering problem is to select a set of $k$ points $C\subseteq S$  to serve all points in $S$ such that $\max_{s\in S}d\left(s,C\right)$ is minimized, where $d\left(s,C\right)=\min_{c\in C}d\left(s,c\right)$ is the distance between a point $s$ and the center set $C$.
Let the data set be divided into $m$ \emph{disjoint} groups $S=S_{1}\cup\ldots\cup S_{m}$. There exists a fairness constraint that the number of chosen centers from group $S_{l}$ is bounded by  $k_l$, where $k_{l}$ is the upper  bound for  group $l$ for all $l \in [m]$  and $[m]$  denotes $\{1,2,\dots,m\}$ for brevity. Moreover,  $\sum^m_{l=1} k_l=k$.
Then,  the fair $k$-center clustering problem is to find a center set $C$ satisfying the formulation as follows: 
\begin{align*}
\min_{C \subseteq S} \quad & \max_{s \in S} \; d(s, C) \\
\text{s.t.} \quad 
& \left| C \cap S_l \right| \le k_l, \quad \forall l, \\
& \left| C \right| \le k.
\end{align*}
\subsection{Algorithmic Framework}
Next, we introduce the general framework of our algorithms. Throughout this paper, we assume that the optimal radius of the fair $k$-center problem is known as $r^*$ (i.e., $r^{*}= \max_{s\in S}d\left(s,C^*\right)$ for the optimum center set $C^*$). Although the exact value of $r^*$ is actually unknown, we can employ the previous approach for finding suitable replacement of $r^*$ in the offline setting due to \citet{guo2025near} by compromising the ratio with a multiplicative factor of $1+\epsilon$. 
In general, all of our algorithms mainly proceed in two stages.

\begin{itemize}

\item  \textbf{Streaming stage}. Select a set of representative data points along the stream according to the optimum radius $r^*$, where the union of the selected point sets might have size larger than $k$; 

\item \textbf{Post-streaming stage}. Compute the set of actual centers from the set of streamed and stored representative data points according to $r^*$.

\end{itemize}

\subsection{The Streaming Stage} 

For the streaming stage, we employ the $\lambda$-independent center set as defined below: 

\begin{definition}
\label{def:1}($\lambda$-independent center set) $\Gamma\subseteq S$ is a $\lambda$-independent center set of $S$ if and only if it satisfies the following two conditions:
\begin{itemize}
    \item[ 1)] For  any two points   $p,q\in \Gamma$, the distance between them is larger than $\lambda$, i.e. $d(p,q)>\lambda$.
    \item[ 2)] For any point $p\in S$, there exists a point  $q\in \Gamma$, such that $d(p,q)\leq \lambda$.
\end{itemize}
\end{definition}

We say a $\lambda$-independent center set  $\Gamma\subseteq S$ is \textit{minimal} iff removal of any point from  $\Gamma$ makes $\Gamma$ no longer a $\lambda$-independent center set.
Assume $C^{*}$ is the center set of an optimal solution and recall that $r^{*}$ is the optimal radius. Then we have:

\begin{lemma}
\label{lem:d-independent center set}
   For a minimal $\lambda$-independent center set $\Gamma\subseteq S$, if $\lambda \geq 2r^{*}$, then $|\Gamma|\leq k$. 
\end{lemma}

\proof{
We only need  to show the case for $\lambda=2r^*$. Suppose $|\Gamma| > k$. Then, by the Pigeonhole principle, there must exist at least a pair of points $i$, $j\in \Gamma$, which are covered by the same center in the optimal solution, say $c^*$ in $C^*$. Then, both $d(i,c^*)\leq r^*$ and $d(c^*,j)\leq r^*$ hold.  By the triangle inequality, we then have $d(i,j)\leq d(i, c^*)+d(c^*,j)\leq r^*+r^*=2r^*$. On the other hand, we have $d(i,j)>2r^*$ according to Cond.~(1) in Def.~\ref{def:1}, resulting in a contradiction and hence completing the proof. 
\qed
}

So for the streaming stage, we only need to construct a \textit{minimal} $\lambda$-independent center set $\Gamma_{l}$ for each $l\in \{1, 2\}$  regarding $\lambda=2r^*$ along the stream:
\begin{mdframed}[linewidth=0.6pt]
\begin{itemize}
    \item[]
\textbf{Upon} \textit{each arriving point} $i$ \textbf{do} 
 \item[]\quad\quad \textbf{if}  $i\in S_l$ and $d\left(i,\Gamma_{l}\right)> \lambda$ both hold \textbf{then} 
 \item[]\quad\quad\quad\quad Grow $\Gamma_l$ (initially empty) by adding  $i$. 
\end{itemize}
\end{mdframed}

\section{Post-streaming Stage of Fair $k$-Center}
\label{sec:5-approximation}

Observing that the size of  $\Gamma_{1}\cup \Gamma_{2}$ obtained in the streaming stage may exceed $k_l$ for $\lambda=2r^*$, we devise a method to construct the actual center set $C$ by selecting points from the union of $\Gamma_{l}$ for $m=2$ based on dealing with the cases whether $\mid\Gamma_{l}\mid \leq k_l$ holds or not. We extend the method to general $m$ in the full version. 

\subsection{Overview of Post-streaming Stage}

Our algorithm for this stage is grounded in analyzing the sizes of $\Gamma_{1}$ and $\Gamma_{2}$. There are exactly three possible cases regarding the sizes of  $\Gamma_{1}$ and $\Gamma_{2}$:
 \begin{itemize}[leftmargin=20pt]
     \item[(1)] $|\Gamma_{1}|\leq k_{1},|\Gamma_{2}|\leq k_{2}$;
     \item[(2)] $|\Gamma_{1}|\leq k_{1},|\Gamma_{2}|>k_{2}$ or $|\Gamma_{1}|> k_{1},|\Gamma_{2}|\leq k_{2}$;  
     \item[ (3)] $|\Gamma_{1}|>k_{1},|\Gamma_{2}|>k_{2}$. 
  \end{itemize}

{For Case~(1)}, we can directly use $C=\Gamma_{1}\cup\Gamma_{2}$ as the desired center set, because:
(1) $|C|= \vert \Gamma_{1}\cup\Gamma_{2} \vert \le k_1+k_2 = k$; (2) for any point $s\in S$, $d(s,C)\leq 2r^*$ holds due to the definition of $\lambda$-independent center set and $\lambda=2r^*$.

Then, we give a method to solve Case (2) and show that Case (3) can be reduced to Case (2).

\subsection{Processing of Case~(2)}

Without loss of generality, we assume that $|\Gamma_{l}|>k_{l}$ and $|\Gamma_{3-l}|\leq k_{3-l}$ for $l\in \{1,2\}$. 
Then, our algorithm simply proceeds as follows: 
\begin{itemize}
    \item [ ] \textit{Construct $\Gamma_{l}'= \{i \vert i\in\Gamma_{l}, d\left(i,\Gamma_{3-l} \right)>3r^{*}\}$ and \\
    return $C=\Gamma_{3-l}\cup \Gamma_{l}'$ as the desired center set.}
\end{itemize}
  The correctness of such $C$ can be derived from the following theorem:
\begin{theorem}\label{thm:case2}
    For $C=\Gamma_{l}'\cup \Gamma_{3-l}$, we have: (1) $|\Gamma_{l}'|\leq k_{l}$, $|\Gamma_{3-l}|\leq k_{3-l}$, $|C|\leq k$; (2) for $\forall s\in S$, $d(s,C)\leq 5r^*$. 
\end{theorem}

\begin{proof}
To prove $|\Gamma_{l}'|\leq k_{l}$ of Cond.~(1), we first show that each  $i\in\Gamma_{l}'$ cannot be covered by any point of $S_{3-l}$ in the optimal clustering,  i.e., $d(i, j)>r^{*}$ holds for any point $j\in S_{3-l}$.  Suppose this is not true, i.e., for a point $i\in\Gamma_{l}'$,  there exists a point $j$ in $S_{3-l}$ with $d(i,j)\leq r^*$. According to the definition of the $2r^{*}$-independent center set, there must exist a point $p\in\Gamma_{3-l}$ with $d(j,p)\leq2r^{*}$. So we get $d(i,p)\leq d(i,j)+d(j,p)\leq r^{*}+2r^{*}=3r^{*}$, where the second inequality is by the triangle inequality.  This contradicts the definition of $\Gamma_{l}'$ that $d\left(i,\Gamma_{3-l}\right)>3r^{*}$ holds for each $i\in \Gamma_{l}'$.    Then, $|\Gamma_{l}'|\leq k_{l}$ must hold, because otherwise there must exist at least two points $p,q\in \Gamma_{l}'$  belonging to an identical cluster in the optimal solution. That means $d(p,q)\leq 2r^*$, contradicting the fact $d(p,q)>2r^*$ as $p,q$ are two points of the $2r^*$-independent center set $\Gamma_l$.

Next, we show $d(s,C)\leq 5r^*$ for any $s\in S$. Firstly,  following  the algorithm,  there exists a point $p\in \Gamma_{l}$ for any $s\in S_{l}$ such that $d(s,p)\leq 2r^{*}$. Then if point $p\in \Gamma_{l}'$, $d(s,p)\leq 2r^{*}$ is true; otherwise, i.e. $p\notin \Gamma_{l}'$, then there exists a point $q\in \Gamma_{3-l}$ with $d(p,q)\leq 3r^{*}$, indicating $d(s,q)\leq d(s,p)+d(p,q) \leq 2r^{*}+3r^{*}=5r^{*}$. Secondly, for each point $s\in S_{3-l}$, there exists a point $p\in \Gamma_{3-l}$ such that $d(s,p)\leq 2r^{*}$. Moreover, $\Gamma_{3-l}$ remains unchanged during the algorithm, so $d(s,p)\leq 2r^{*}$ remains true for $p\in \Gamma_{3-l}$. Hence, for each $s\in S=S_{1}\cup S_{2}$, $d(s,C)\leq 5r^*$ holds.
\end{proof}

\subsection{Processing of Case~(3)} 

 We propose a more sophisticated algorithm to ensure the approximation ratio of $5$ for Case~(3).
The key idea of our algorithm is to construct an auxiliary bipartite graph $G(\Gamma_1\cup \Gamma_2, E)$, such that the center selection problem regarding $\Gamma_1\cup \Gamma_2$ is transformed into a special constrained vertex cover problem in $G$.
The construction of $E(G)$ simply proceeds as:
\begin{itemize}
    \item [] \textit{For  every pair of points $p\in \Gamma_1$ and $q\in \Gamma_2$,  add an\\ edge  $e(p,q)$ to  $E(G)$ {if and only if} $d(p,q)\leq 3r^*$.}
\end{itemize}

Then, the aim is to find a vertex cover in $G$  with a size bounded by $k$ that also satisfies the fairness constraints $k_1$ and $k_2$. 
Notably, although finding a vertex cover is $\mathcal{NP}$-hard in general, we manage to devise a polynomial-time exact algorithm for the problem based on certain special properties of the constructed auxiliary graph.

The key idea of our algorithm is to repeatedly eliminate degree 0 and degree 1 points in $G$ when any exist, by selecting the point covering such points from $V(G)$ as a new center in the set $C$. While $G$ contains no degree-0 or degree-1 points, our algorithm arbitrarily chooses an edge in $G$, and selects one of its endpoints as a new center by adding it to $C$. The procedure repeats until the problem reduces to Case (2). 
The detailed algorithm is illustrated in Alg.~\ref{alg:3}.

\begin{algorithm}[t]
  \small
 \caption{Processing of Case~(3)}
 \label{alg:3}
 \KwIn{
 $\Gamma_{1}$ and $\Gamma_{2}$ with $|\Gamma_{1}|>k_{1}$, $|\Gamma_{2}|> k_{2}$.
 }  
 \KwOut{Center set $C$.}
  \tcp{Phase 1: Construction of $G$ and~$C$.}
 Set $C\leftarrow \emptyset$\;
Construct the auxiliary graph $G(\Gamma_1\cup\Gamma_2, E)$, where an edge $e(p,q)$ exists in $E$ if and only if  $p\in \Gamma_{1}$, $q\in \Gamma_{2}$ and  $d(p,q) \leq 3r^*$\;
 \For{each point $i\in G$ with degree 0}{
    \If {$d(i,C) > 2r^*$}{
    Set $C\leftarrow C\cup\left\{ i\right\} $\;
    }
} 
Remove each point $i\in G$ with degree 0\;
\tcp{Phase 2: Final computation of $C$.}
 \While{$|C|\leq k=k_1+k_2$ and $V(G)\neq \emptyset$}{
\eIf{there exists no degree-1 point in $G$}{Arbitrarily select an edge $e(p,q)$ and remove it from $G$\;
Set  $C\leftarrow C\cup\left\{ p\right\}$ and $G\leftarrow G\setminus\{p,\,q\}$\;
Set  $\Gamma_l\leftarrow \Gamma_l \setminus\left\{ p\right\}$ and $\Gamma_{3-l}\leftarrow \Gamma_{3-l} \setminus\left\{ q\right\}$ for $\Gamma_l$ containing ${p}$\;
}{
 Find   $i\in G$ with largest $N_1(i)$,  where $N_1(i)$ is the set of degree-1 neighbours of $i$ in $G$\;
 Set $C\leftarrow C\cup\left\{ i\right\}$ and  $G\leftarrow G\setminus \{i\} \setminus N_1(i)$\;
 Set $\Gamma_{l}\leftarrow \Gamma_{l}\setminus\left\{ i\right\}$ and  $\Gamma_{3-l}\leftarrow \Gamma_{3-l}\setminus N_1(i)$ for  $\Gamma_l$ containing $i$\;
}

\If{ there exists $l$ with $|C\cap S_l| + |\Gamma_l| \leq k_l$  }{ 

Call \textit{the algorithm for Case (2)} to find the center set $C'$ on $\Gamma_l\cup \Gamma_{3-l}$ with $k'_l=k_l-|C\cap S_l|$.

Return  $C\cup C'$.
}  

}

\end{algorithm}

\begin{lemma}
\label{lem:the property of Gamma}
     Each pair of points from the same $\Gamma_l$ belongs to different clusters of the optimum solution when $\lambda = 2r^*$ before commencing of  Alg.~\ref{alg:3}.
\end{lemma}
\begin{proof}
When $\lambda=2r^*$, according to Def.~\ref{def:1}, the distance between each pair of points in $\Gamma_l$ is larger than $2r^*$. Suppose there exists a pair of points $i,j$ in $\Gamma_l$ such that they are in the same cluster of the optimum solution. That is, there exists a center $c^*$ in the center set of the optimum solution such that $d(i,c^*)\leq r^*$ and $d(j,c^*)\leq r^*$ both hold. So $d(i,j)\leq d(i,c^*)+d(j,c^*)\leq r^*+r^*=2r^*$, contradicting $d(i,j)>2r^*$.
\end{proof}

\begin{lemma}
\label{lem:fairness holds}
When Phase 1  of Alg.~\ref{alg:3} completes, $|C\cap S_l|\leq k_l$ holds for the current center set $C$. 
\end{lemma}
\begin{proof}
{Following Steps 3-6 in Alg.~\ref{alg:3}, we add $i\in G$ with degree $0$ to $C$.} That is, each point $i\in \Gamma_l$ added to $C$ satisfies $d(i,\Gamma_{3-l})>3r^*$. {Then for any point $j\in S_{3-l}$, we have $d(j,\Gamma_{3-l})\le 2r^*$, and hence $d(i,j)>r^{*}$ holds.} That is, each point $i\in C\cap \Gamma_l$ cannot belong to a cluster centered at a point of $S_{3-l}$ in the optimal solution. On the other hand, each pair of points $i,j\in C$ belongs to different optimal clusters because $d(i,j)>2r^*$ (similar to the proof of Lem.~\ref{lem:the property of Gamma}). Therefore, we have $|C\cap S_l|\leq k_l$ for every $l$.
\end{proof}

For brevity,  $\Gamma^j_l$ and $C^j$ respectively denote $\Gamma_l$ and $C$ in the $j$th iteration of the while-loop  (Phase 2) of Alg.~\ref{alg:3}.
\begin{lemma}
\label{lem:last fairness constraints}
    In the $j$th iteration of the \textit{while} loop of Alg.~\ref{alg:3},  we have:  (1) $|\Gamma^j_l|\leq k-|C^j|$;  (2) $|C^j|\leq k$; (3) $|C^j\cap S_l|\leq k_l$.
\end{lemma}
\begin{proof}
    For Cond.~(1), we first show it is true when $j=1$. Before the first while loop commences,  $|\Gamma^1_l|\leq k-|C^1|$ is satisfied. Suppose otherwise,  there would exist a contradiction as analyzed below. We will show the distance between every two points belonging to $C^1\cup \Gamma_l^{1}$ is larger than $2r^*$, so every two points therein cannot belong to the same cluster in an optimal solution,  and hence there are at least  $|\Gamma^1_l|+|C^1|>k$ centers in an optimal solution, leading to a contradiction.  This can be deduced following Steps 3-6 in Alg.~\ref{alg:3}, where points with a zero degree are added to $C^1$ for the first iteration. For each pair of  points $i\in C^1$ and $j\in \Gamma_l^{1}$, the absence of  edge $e(i,j)$ in $G$ indicates $d(i,j)>3r^*$ according to the construction of the auxiliary graph $G$. Moreover, for each two points $i,j\in \Gamma_l^{1}$, $d(i,j)>2r^*$ holds because  $\Gamma_l$ is a $2r^*$-independent center set for $S_l$ ($l\in \{1,2\}$), and $i,j$ are points in initial $\Gamma_l$. Lastly, for each  $i,j\in C^{1}$, $d(i,j)>2r^*$ holds because of Step 4 of the algorithm. Therefore, $d(i,j)>2r^*$ holds for each two points $i,j\in C^1\cup\Gamma_l$. 

We demonstrate that Cond.~(1) holds for the $(j+1)$th iteration by induction. Assume that Cond.~(1) is valid for the $j$th iteration concerning $\Gamma^j$. Then, we either remove two points or remove a point $i$ along with its degree-1 neighbors. 
Consequently, after this removal, it follows that:
\[|\Gamma^{j+1}_l| \leq |\Gamma^{j}_l| - 1 \leq k - |C^{j}| - 1 \leq k - |C^{j+1}|,\]
where the first inequality is due to the removal of at least one point from both $\Gamma^j_1$ and $\Gamma^j_2$ in every iteration. The second inequality is derived from our inductive assumption. The third inequality arises because, on one hand, we add one point to $C^{j}$ to form $C^{j+1}$, and on the other hand, according to the algorithm, the removal of points from $G$ does not result in new points of degree 0 in the auxiliary graph $G$. Hence, Cond.~(1) is also true for the $(j+1)$th iteration. 

Cond.~(2) can be immediately derived from Cond.~(1) because $|\Gamma^{j}_l|\geq 0$ holds. 

Moving on to Cond.~(3), Lem.~\ref{lem:fairness holds} establishes that in Steps 3-6 of Alg.~\ref{alg:3}, the inequality $|C^1\cap S_l|\leq k_l$ is true. That is, Cond.~(3) holds for $j=1$. Assume that Cond.~(3) holds for $j$th iteration. We need only to verify that this inequality holds in the $(j+1)$th iteration of Alg.~\ref{alg:3}.  According to Steps 8-15, Alg.~\ref{alg:3} involves adding a single point to $C$ in each iteration. So if $|C^j\cap S_l| < k_l$ holds for any $l\in\{1,2\}$, then $|C^{j+1}\cap S_l| \le k_l$. Otherwise (i.e., $|C^j\cap S_l| = k_l$),  we have $|\Gamma^{j}_{3-l}| + |C^{j}|  \leq k$ by  Cond.~(1). This implies
\[|\Gamma^{j}_{3-l}| + |C^{j}\cap S_{3-l}|  \leq k - |C^{j}\cap S_{l}| 
    = k-k_l = k_{3-l}. \]
According to  the \textit{if} condition in Step 16 of Alg.~\ref{alg:3}, when the above  inequality holds,  Alg.~\ref{alg:3} terminates and reduces to Case (2). This completes the proof. 
\end{proof}

\begin{theorem}\label{thm:the last fairness constraints}
     Alg.~\ref{alg:3}  terminates in time $O(k^2)$ and outputs a feasible solution $C$  to fair $k$-center, for which  $d(s,C)\leq 5r^*$ holds for any point  $s\in S$.
\end{theorem}
    
\begin{proof}
For the runtime,  Step 2 of Alg.~\ref{alg:3} takes $O(k^2)$ time to construct $G$ as $G$ contains at most $O(k)$ vertices and $O(k^2)$ edges. Then, the while-loop (Steps 7-18)  iterates for at most   $O(k)$ times as it removes at least one point from $G$ in each iteration, and each iteration takes  $O(k)$ time. Therefore, the total runtime sums up to $O(k^2)$.

From Cond.~(2) and (3) of Lem.~\ref{lem:last fairness constraints}, we immediately get $|C|\leq k$ and $|C\cap S_l|\leq k_l$, which indicates $C$ is a feasible solution. It remains to bound the distance from any point  $s\in S$ to $C$. For any $s\in S_l$,  a point $i\in \Gamma_l$ must exist such that $d(s,i)\leq 2r^*$ holds according to the definition of  $\Gamma_l$. If $i\in C$ holds, then we have  $d(s,C)\leq 2r^*$. Otherwise, i.e. $i\notin C$, then there must exist an edge $e(i,j)$ with $j\in  \Gamma_{3-l}$ in the auxiliary graph. According to Steps 8-15 in Alg.~\ref{alg:3}, at least one endpoint of $e(i,j)$, either $i$ or $j$, must be added to $C$. Since  $i\notin C$ by assumption, $j\in C$ holds. Moreover, the existence of edge $e(i,j)$ also means $d(i,j)\leq 3r^*$. So by triangle inequality, we have 
\[d(s,C)\leq d(s,j)\leq d(s,i)+d(i,j)\leq2r^{*}+3r^{*}=5r^{*}.\]
Therefore, regardless of whether $i\in C$ holds or not, $d(s,C)\leq 5r^{*}$ holds for any $s\in S_l$ for any $l\in\{1,2\}$.\end{proof}

Combining Cases (1), (2) and (3) with the streaming stage, we eventually achieve a complete streaming algorithm for fair $k$-center provided that $r^*$ is known. For the space complexity, the algorithm stores $m=2$  independent center sets for each case where each independent center set has $O(k)$ points, so it consumes a space complexity of $O(k)$.   However, as $r^*$ is actually unknown, we need an $O(\log \alpha)$ multiplicative factor over the memory complexity, where $\alpha$ is the aspect ratio that is the ratio between the maximum and minimum distance between the points of $S$.  
Combining Thm.~\ref{thm:the last fairness constraints}, we have:

\begin{corollary}
When $r^*$ is unknown, fair $k$-center admits a streaming algorithm  with an approximation ratio $5+\epsilon$ and a memory complexity $O(k\log \alpha)$. 
\end{corollary}

Notably, we can extend the algorithms respectively for Cases~(2) and Cases~(3) and obtain an algorithm for general $m$ (See the extended algorithm in the full version). 

\section{ Streaming Semi-structured Datasets}

In this section, we consider scenarios in which the data points are streamed following the demographic group order, i.e. all points belonging to $S_{l}$ arrive before $S_{l+1}$, $\forall l\in[m-1]$. We devise an algorithm for $m=2$, achieving a ratio of $3+\epsilon$ that almost matches the state-of-the-art ratio of $3$ for the offline setting.
The algorithm can be extended to achieve a ratio $4$ for general $m$ as shown in the full version.

Our key idea of the improved algorithm for $m=2$ is to obtain an extra point set $\Gamma_{sub}$ in the streaming stage and use some of its points as centers for  points in $\Gamma_1$. 
More precisely, we first  construct an independent set $\Gamma_1'=\Gamma_1$ with $\lambda=2r^*$ upon the data stream of $S_1$. Then, upon the stream of $S_2$, we select points using different designated approaches depending on whether $\vert \Gamma_1' \vert \le k_1$ holds. As the constructed independent center sets are slightly different, we adopt notation $\Gamma_l'$ instead of  $\Gamma_l$. 
The detailed algorithm is as in Alg.~\ref{alg:3-ratio_for_metric_space}.

Then we prove that with only the representative points of $\Gamma_1' \cup \Gamma_2'$ and $\Gamma_{sub}$, our post-streaming algorithm can compute an approximation solution with ratio $3$.

\begin{algorithm}[t]
 \small
 \setcounter{AlgoLine}{0}
 \SetCommentSty{itshape}
 \caption{Streaming semi-structured data} 
\label{alg:3-ratio_for_metric_space} 
\KwIn{A stream of points $S=\bigcup_{l}S_{l}$ in which   all points of $S_{1}$ arrive before the points of $S_{2}$, and $\lambda=2r^{*}$.}
\KwOut{The sets of representative points.}

Set $\Gamma_{sub}=\emptyset$ and $\Gamma_{l}'=\emptyset$ for  $l=1,2$\;
\Upon{each arriving  point $i\in S_1$}{
    \If{$d\left(i,\Gamma_{1}'\right)> \lambda$}{ 
    Set $\Gamma_{1}' \leftarrow \Gamma_{1}'\cup\left\{ i\right\}$ \tcp*{Recall that $d\left(i,\emptyset\right) = \infty$.}   
    }    
} 
    \Upon{each arriving  point $i\in S_{2}$}{
    \eIf{$|\Gamma_1'|\leq k_1$}{ 
        \If{$d\left(i,\Gamma_{1}'\right)> 3\lambda/2$ and $d(i,\Gamma_{2}')> \lambda$}{
            Set $\Gamma_{2}' \leftarrow \Gamma_{2}'\cup\left\{ i\right\}$\; 
        }
    }{ 
    \If{ $d\left(i,\Gamma_{1}'\cup \Gamma_2' \right)> \lambda$}{ 
        Set $\Gamma_{2}' \leftarrow \Gamma_{2}'\cup\left\{ i\right\}$\; 
    }     
     \If{there exists a point $j\in \Gamma_{1}'$ that has no replacement in $\Gamma_{sub}$ and $d(i,j)\leq \lambda/2$}{ 
        Set $\Gamma_{sub} \leftarrow \Gamma_{sub}\cup\left\{ i\right\}$\; 
    }    
    } 
    }    
   Return $\Gamma_{1}'$, $\Gamma_2'$, and $\Gamma_{sub}$.
\end{algorithm}

\begin{lemma}\label{lem:sizeofc}
    For $\Gamma_{1}'$ and $\Gamma_2'$ produced by Alg.~\ref{alg:3-ratio_for_metric_space}, we have: (1) $|\Gamma_{1}'|+|\Gamma_2'|\leq k$;  (2) $|\Gamma_2'|\leq k_2$; (3) There exists a subset $\Gamma_1''\subseteq \Gamma_1'$ and accordingly a subset $\Gamma'_{sub}=\{\sigma(c)\vert c\in \Gamma_1''\}\subseteq \Gamma_{sub}$, such that  $|\Gamma_1'\setminus \Gamma_1'' \cup \Gamma'_{sub} \cap S_1|\leq k_1$ holds, where $\sigma(c)$ is the stored replacement of $c$ in $\Gamma_{sub}$ with $d(c,\sigma(c))\leq r^{*}$.
\end{lemma}

\begin{proof}
    For Cond.~(1), according to the algorithm,  we have $d(p,q)>2r^*$ for any pair of  two points $p,q\in \bigcup\Gamma_{l}', l=1,2$. So every pair of points of $\bigcup\Gamma_{l}'$ must appear  in two different clusters in the optimal solution similar to the proof of Lem~\ref{lem:the property of Gamma}. Thus, $|\bigcup_{l=1}^2\Gamma_{l}'|\leq k$ holds.

    For  Cond.~(2), if $|\Gamma_1'|> k_1$, we have $|\Gamma_2'|\leq k-|\Gamma_1'| < k-k_1=k_2$ . Otherwise (i.e., $|\Gamma_1'|\le k_1$), as $\Gamma_1'$ covers all points of the clusters centered at points of $S_1$ in the optimal solution within a radius $3r^*$. Then, the remaining points to be covered in the algorithm appear in only the optimal clusters centered at points of $S_2$, which are at most $k_2$ clusters. On the other hand, every pair of points in $\Gamma_2'$ must appear in different optimal clusters. So $|\Gamma_2'|\le k_2$.

      For Cond.~(3),   if $|\Gamma_1'|\leq k_1$ holds after streaming then the proof is done. So we need only to prove the case for $|\Gamma_1'|> k_1$. Clearly, there exists a subset of $\Gamma_1'$ that contains at least $|\Gamma_1'| - k_1$ points, say $\Gamma_1''$, appearing in optimal clusters centered at points of $S_2$.  Clearly, there exists at least one point  $j\in S_2$ for  each point  $i\in \Gamma_1''$ with $d(i,j)\leq r^*$. Moreover, according to Alg.~\ref{alg:3-ratio_for_metric_space},  we add such point $\sigma(i)=j$ to $\Gamma_{sub}$ as a replacement point of $i$, collectively composing the set $\Gamma'_{sub}$. That is, $\Gamma_1'\setminus \Gamma_1'' \cup \Gamma'_{sub}$ feasibly covers all points of $\Gamma'_1$ with radius $r^*$, and contains at most  $k_1$ points of $\Gamma_1'$. 
\end{proof}
Following the above lemma, we can immediately obtain a 3-approximation solution by simply using $C=\Gamma'_2 \cup (\Gamma_1'\setminus \Gamma_1'' \cup \Gamma'_{sub})$ as the desired center set. 

\section{Experimental Results}
\label{sec:experiment}

In this section, we conduct an empirical evaluation of our algorithms utilizing both simulated and real-world datasets, compared with three previous approximation algorithms that serve as baselines. All experiments are averaged over at least 20 iterations and are conducted on a Linux machine equipped with a 12th Gen Intel(R) Core(TM) i9-12900K CPU and 32 GB of RAM using Python 3.8. A detailed description of the experimental setup is provided in follows.

\subsection{Experimental Setting}
\label{subsec:Experiment}

\paragraph{Datasets.} We employ both simulated and real-world datasets to evaluate our approximation algorithm. We summarize the detailed datasets in Tab.~\ref{tab:datasets}.

\textbf{\textit{Real-world datasets}}. We apply our algorithms to {seven real-world datasets}: Wholesale~\cite{wholesale_customers_292}, Student~\cite{student_performance_320}, Bank~\cite{bank_marketing_222}, CreditCard~\cite{default_of_credit_card_clients_350}, Adult~\cite{adult_2}, SushiA and CelebA~\cite{liu2015faceattributes}, following the most recent related works~\cite{jones2020fair,chen2019proportionally}. Consistent with previous studies, we utilize meaningful numerical features for clustering and incorporate selected binary categorical attributes to construct datasets with fair constraints across all datasets.

\textbf{\textit{Simulated datasets}}. We provide two datasets that serve as benchmarks: the simulated dataset provides an exact optimal solution, while the large-scale dataset incurs more file I/O overhead, making it ideal for measuring computational efficiency. First, the dataset (called \textit{SimuA}) is used to assess the \textit{empirical approximation ratio} of our algorithms and the baseline methods. Inspired by previous research~\cite{kleindessner2019fair}, we use their method to construct a simulated dataset with a known optimal solution for the fair $k$-center problem. Secondly, we leverage the implementation from~\citet{chiplunkar2020solve} to generate a \textit{100 GB} dataset (called \textit{SimuB}), which allows us to evaluate the {runtime} performance of different algorithms. 

\paragraph{Constraint Settings.} For the simulated dataset, we configured the parameters to compare both the average and maximum approximation ratios for these algorithms and to verify the approximation ratio of our algorithms. To ensure fair center selection, we aligned the number of required centers from each group to be proportional to the size of that group. Following the fairness principle of disparate impact as outlined by~\citet{feldman2015certifying}, we restricted the selection to $k_l$ data points from the $l$th group to serve as centers. We then evaluated the clustering quality by varying the number of $k$ clusters across these datasets.

\begin{figure*}[t]
  \centering
  \includegraphics[width=\textwidth]{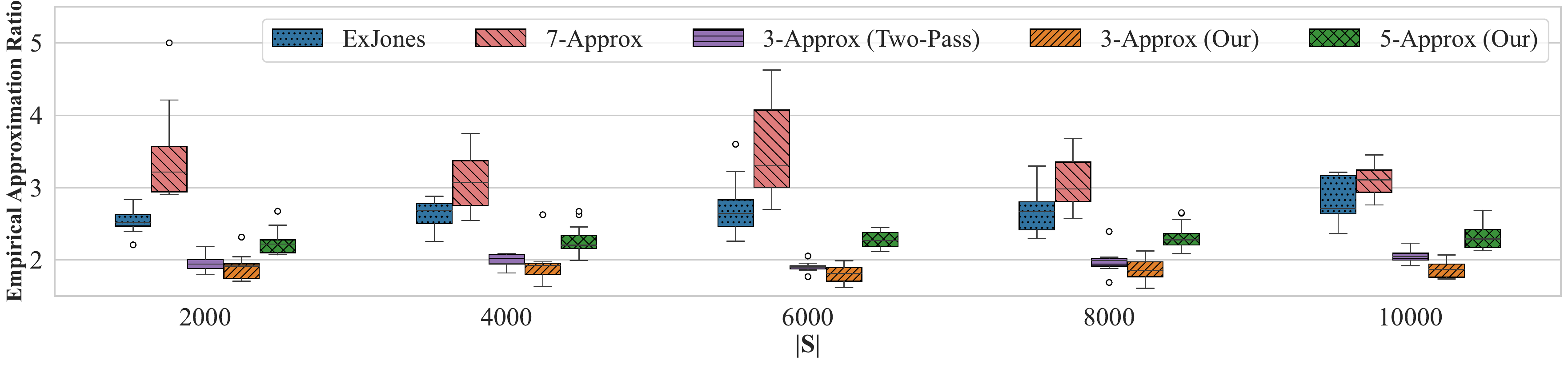} 
  \caption{Empirical approximation ratio ($cost/r^*$) of our algorithms in comparison with other baselines.}
  \label{fig:exp_approx}
\end{figure*}
\paragraph{Algorithms.} In comparison with existing streaming fair $k$-center clustering algorithms, this experiment includes an extended algorithm from the offline 3-approximation algorithm by~\citet{jones2020fair} (denoted as \textit{ExJones}), a two-pass streaming 3-approximation algorithm by~\citet{chiplunkar2020solve}  (denoted as two-pass \textit{$3$-Approx}), and a one-pass streaming 7-approximation algorithm~\cite{lin2024streaming} (denoted as \textit{$7$-Approx}). Moreover, we propose two algorithms in this paper: 1) a one-pass 5-approximation algorithm (denoted as \textit{$5$-Approx}); 2) a one-pass 3-approximation algorithm with $m = 2$ (denoted as \textit{$3$-Approx}) for semi-structured data streams in our experiments.

\descr{Metrics}.
To evaluate clustering quality, we use the \textit{cost} metric defined in Sec.~\ref{sec:prob_statement} and compare its average value across datasets. We also report runtime, measured in seconds.

\begin{table}[t]
\centering
{\footnotesize
\begin{tabularx}{\linewidth}{lXl}
\toprule[1pt]
\textbf{Dataset} & \textbf{\#Records} & \textbf{\#Dim.} \\
\midrule[0.8pt]
Wholesale  & 440 & 6 \\
Student & 649 & 16 \\
Bank & 4,521 & 7 \\
SushiA & 5,000 & 10 \\
CreditCard & 30,000 & 19 \\
Adult & 32,561 & 6 \\
CelebA & 202,599 & 15,360 \\
\midrule[0.4pt]
SimulatedA & 4,000,000& 1,000 \\
SimulatedB & [2k,4k,6k,8k,10k] & 2\\
\midrule[1pt]
\end{tabularx}
}
\caption{Datasets Summary.\label{tab:datasets}}
\end{table}

\subsection{Experimental Analysis}
\label{subsec:Exper}

In Fig.~\ref{fig:exp_approx}, we first report the empirical approximation ratios of the algorithms on a simulated dataset with known optimal solutions. This allows us to compare the experimental objective values of our algorithms and the baselines against the optimal objective. We then evaluate the clustering costs of these algorithms on real-world datasets. Finally, we report the runtime(s) performance on large-scale simulated datasets, focusing on file I/O operations in our algorithms compared to other baselines.

\paragraph{Approximation Factor.}
We compare algorithm performance by evaluating the relative solution ratio against the optimal solution provided for the simulated dataset.  The ratio of the evaluation result can be called the empirical approximation ratio, and the maximum value represents the worst-case cost in all the experimental results. We run the code of constructing the simulated dataset shared by~\citet{kleindessner2019fair} for their algorithm setting $m = 2$, $k = 100$ and varying the size of the simulated datasets  $|S|$ from $2,000$ to $10,000$. For each size of the dataset, we perform $10$ runs on $10$ kinds of fairness with $10$ different constraint ratios.

\begin{table}[t]
\centering
\scalebox{0.75}{
\begin{tabular}{l|c|c|c||c||c}
\toprule[1pt]
\textbf{Dataset} & ExJones & $7$-Approx & $5$-Approx & \makecell{$3$-Approx\\(\textbf{Semi})} & \makecell{$3$-Approx\\(\textbf{Two pass})} \\
\midrule[0.8pt]
Wholesale & 1.36 & 1.28 & \underline{1.04} & \textbf{0.84} & {0.61} \\
Student   & \underline{1.85} & 1.85 & 1.90 & \textbf{1.79} & {1.73} \\
CelebA    & 24.36 & 26.27 & \underline{\textbf{21.34}} & \textbf{21.34} & {21.25} \\
Bank      & \underline{0.49} & 0.93 & 0.61 & \textbf{{0.40}} & 0.43 \\
SushiA    & 1.42 & 1.66 & \textbf{\underline{1.41}} & 1.51 & 1.48 \\
\midrule[0.5pt]
Credit    & \underline{\textbf{0.93}} & 1.89 & 1.03 & 0.94 & {0.81} \\
Adult     & 0.65 & 1.08 & \underline{0.62} & \textbf{{0.52}} & 0.56 \\
\midrule[1pt]
\end{tabular}
}
\caption{Cost comparison on the real-world datasets. (The \underline{underline} values highlight the best results of the general one-pass algorithms. In addition, the bold indicates the best result when semi-structured data streaming is included alongside the one-pass streaming algorithms.)\label{tab:exp_cost_real}}
\end{table}

We observe that the practical performance of all algorithms aligns well with their theoretical bounds. Notably, the $3$-Approx (semi) demonstrates the best performance, even outperforming the two-pass 3-Approx algorithm~\cite{chiplunkar2020solve}. 
When the dataset size $|S|=2,000$, the maximal empirical approximation ratio of 7-Approx is greater than 5, validating the suitability of the dataset for testing these algorithms. Compared with the method of Jones et al., our algorithm achieves a lower empirical approximation ratio, demonstrating the advantage of our algorithm. 
Consequently, our algorithms exhibit superior performance on the datasets, conforming to the theoretical guarantees.

\subsubsection{Comparison with Baselines.}
 We report the clustering costs achieved by our algorithms and three baselines across seven real-world datasets in Tab.~\ref{tab:exp_cost_real}. The parameter $k$ is set proportionally to the dataset size: datasets above the dividing line use $1\%$ of the total number of records, while those below use $1\text{\textperthousand}$, which depends on the variation in dataset scales. This same proportion is also applied to the center constraints across different groups. All datasets are normalized prior to running the algorithms, and for CelebA, the first 1,000 samples are used in the experiments.

In Tab.~\ref{tab:exp_cost_real}, we observe that $3$-Approx (Semi) outperforms most of the other one-pass algorithms, while the two-pass $3$-approximation algorithm maintains better performance than the one-pass streaming algorithms on the seven datasets.
We attribute this to the fact that our algorithm is a modification of the original center selection method based on streaming. By combining streaming techniques with the structure of an \textit{independent center set}, it effectively identifies more meaningful center points. As a result, it achieves lower clustering costs than the other baselines under the one-pass streaming setting (shown as the underline values).
We also observe that the results reported by (two-pass) $3$-Approx consistently outperform all the one-pass streaming algorithms. This can be explained by the fact that the two-pass 3-approximation algorithm has the advantage of refining the center set in the second pass, enabling it to identify more accurate center points.

\begin{figure}[t]
  \centering
  \includegraphics[width=\linewidth]{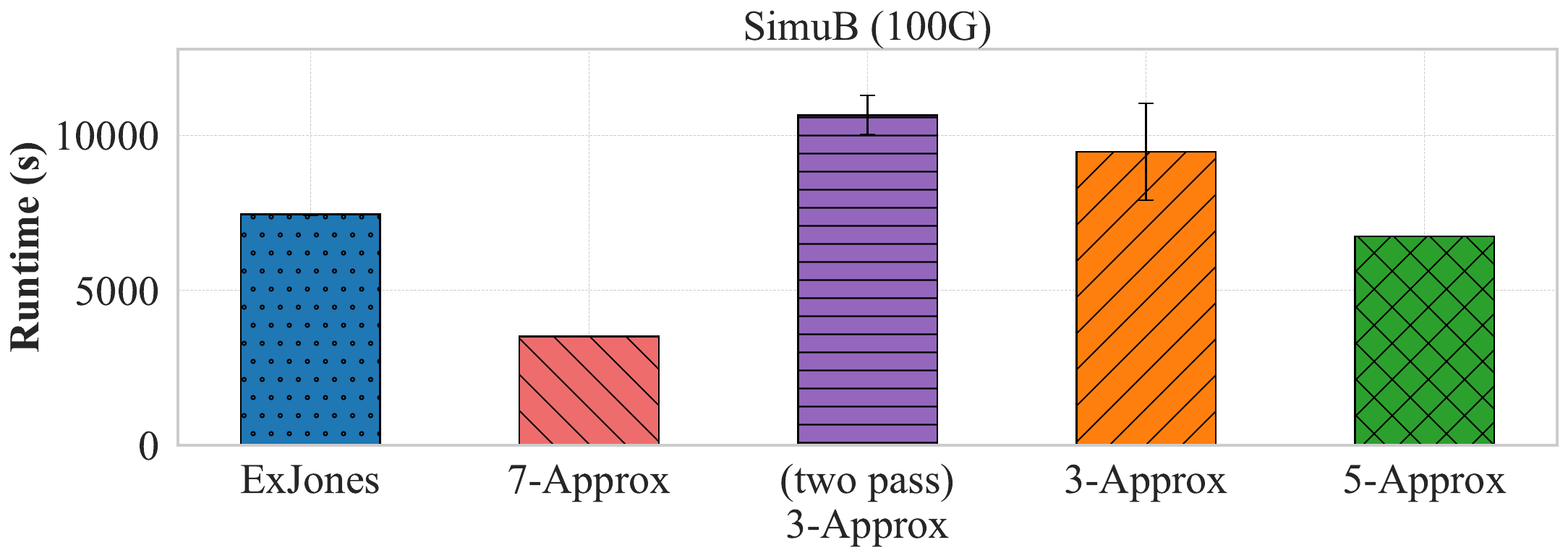} 
  \caption{Runtime on the 100G simulated dataset.\label{fig:exp_runtime}}
\end{figure}

\subsubsection{Streaming Runtime.}
Using the method from~\citet{chiplunkar2020solve}, we generate a 100 GB dataset containing 4,000,000 items and 1,000 features. In Fig.~\ref{fig:exp_runtime}, we evaluate the runtime performance of our algorithms and baseline methods on this dataset.
Our results show that the $5$-Approx algorithm runs faster than the algorithms of~\citeauthor{chiplunkar2020solve} and \citet{jones2020fair}, which aligns with theoretical expectations, as all of our algorithms are one-pass streaming methods. 
Among these methods, 7-Approx~\cite{lin2024streaming} runs the fastest, because it selects fewer points during the streaming process and omits post-processing steps, thereby reducing computational overhead.

\section{Conclusion}
In this paper, we first devise a one-pass streaming algorithm with an approximation ratio of $5+\epsilon$ and a memory complexity of $O(k \log \alpha)$, where $\alpha$ is the aspect ratio. Motivated by the broad applications of semi-structured data streams, we then present an algorithm with a ratio of $(3+\epsilon)$ and the same memory complexity for the problem on semi-structured streams with $m=2$. This $(3+\epsilon)$-approximation can be directly used to solve the offline problem and achieve a slightly improved ratio of 3, matching the state-of-the-art ratio for offline fair $k$-center. Lastly, we conduct extensive experiments to demonstrate the performance gains of our algorithms compared to baselines, including state-of-the-art methods. In the future, we will strive to design both offline and streaming algorithms with better approximation ratios for fair $k$-center, $k$-median, and $k$-means in doubling metrics and Euclidean spaces. 

\section*{Acknowledgments}

This work is supported by National Natural Science Foundation of China (No. 12271098) and Key Project of the Natural Science Foundation of Fujian Province (No. 2025J02011). 

\bibliography{reference}

@inproceedings{kleindessner2019fair,
  title={Fair $k$-center clustering for data summarization},
  author={Kleindessner, Matth{\"a}us and Awasthi, Pranjal and Morgenstern, Jamie},
  booktitle={Proceedings of the 36th International Conference on Machine Learning (ICML-19)},
  pages={3448--3457},
  year={2019},
  organization={PMLR}
}

@article{chen2016matroid,
  title={Matroid and knapsack center problems},
  author={Chen, Danny Z and Li, Jian and Liang, Hongyu and Wang, Haitao},
  journal={Algorithmica},
  volume={75},
  pages={27--52},
  year={2016},
  publisher={Springer}
}

@article{jia2022fair,
  title={Fair colorful $k$-center clustering},
  author={Jia, Xinrui and Sheth, Kshiteej and Svensson, Ola},
  journal={Mathematical Programming},
  volume={192},
  number={1-2},
  pages={339--360},
  year={2022},
  publisher={Springer}
}

@inproceedings{anegg2022techniques,
  title={Techniques for Generalized Colorful $k$-Center Problems},
  author={Anegg, Georg and Vargas Koch, Laura and Zenklusen, Rico},
  booktitle={Proceedings of the 30th Annual European Symposium on Algorithms (ESA-22)},
  year={2022},
  organization={Schloss Dagstuhl-Leibniz-Zentrum f{\"u}r Informatik}
}

@inproceedings{angelidakis2022fair,
  title={Fair and Fast $k$-Center Clustering for Data Summarization},
  author={Angelidakis, Haris and Kurpisz, Adam and Sering, Leon and Zenklusen, Rico},
  booktitle={Proceedings of the 39th International Conference on Machine Learning (ICML-22)},
  pages={669--702},
  year={2022},
  organization={PMLR}
}

@article{bandyapadhyay2019constant,
  title={A Constant Approximation for Colorful $k$-Center},
  author={Bandyapadhyay, Sayan and Inamdar, Tanmay and Pai, Shreyas and Varadarajan, Kasturi},
  journal={Leibniz International Proceedings in Informatics, LIPIcs},
  volume={144},
  year={2019},
  publisher={Schloss Dagstuhl-Leibniz-Zentrum fuer Informatik GmbH, Wadern/Saarbruecken~…}
}

@inproceedings{jones2020fair,
  title={Fair $k$-centers via maximum matching},
  author={Jones, Matthew and Nguyen, Huy and Nguyen, Thy},
  booktitle={Proceedings of the 37th International Conference on Machine Learning (ICML-20)},
  pages={4940--4949},
  year={2020},
  organization={PMLR}
}

@article{kale2019small,
  title={Small Space Stream Summary for Matroid Center},
  author={Kale, Sagar},
  journal={Approximation, Randomization, and Combinatorial Optimization. Algorithms and Techniques (APPROX/RANDOM 2019)},
  volume={145},
  pages={20},
  year={2019},
  publisher={Schloss Dagstuhl--Leibniz-Zentrum fuer Informatik}
}

@article{matthew2008streaming,
  title={Streaming Algorithms for $k$-Center Clustering with Outliers and with Anonymity},
  author={Matthew McCutchen, Richard and Khuller, Samir},
  journal={Approximation, Randomization and Combinatorial Optimization. Algorithms and Techniques},
  pages={165--178},
  year={2008},
  publisher={Springer}
}

@article{nguyen2022fair,
  title={Fair Range $k$-center},
  author={Nguyen, Huy L{\^e} and Nguyen, Thy and Jones, Matthew},
  journal={arXiv preprint arXiv:2207.11337},
  year={2022}
}

@inproceedings{feldman2015certifying,
  title={Certifying and removing disparate impact},
  author={Feldman, Michael and Friedler, Sorelle A and Moeller, John and Scheidegger, Carlos and Venkatasubramanian, Suresh},
  booktitle={proceedings of the 21st ACM SIGKDD international conference on knowledge discovery and data mining (KDD-15)},
  pages={259--268},
  year={2015}
}

@inproceedings{chierichetti2017fair,
  title={Fair clustering through fairlets},
  author={Chierichetti, Flavio and Kumar, Ravi and Lattanzi, Silvio and Vassilvitskii, Sergei},
  booktitle={Proceedings of the 31st International Conference on Neural Information Processing Systems (NeurIPS-17)},
  pages={5036--5044},
  year={2017}
}

@inproceedings{chiplunkar2020solve,
  title={How to solve fair k-center in massive data models},
  author={Chiplunkar, Ashish and Kale, Sagar and Ramamoorthy, Sivaramakrishnan Natarajan},
  booktitle={Proceedings of the 37th  International Conference on Machine Learning (ICML-20)},
  pages={1877--1886},
  year={2020}
}

@article{ceccarello2019solving,
  title={Solving k-center Clustering (with Outliers) in MapReduce and Streaming, almost as Accurately as Sequentially},
  author={Ceccarello, Matteo and Pietracaprina, A and Pucci, G},
  journal={Proceedings of the VLDB Endowment},
  volume={12},
  number={7},
  pages={766--778},
  year={2019},
  publisher={VLDB Endowment}
}

@inproceedings{chen2019proportionally,
  title={Proportionally fair clustering},
  author={Chen, Xingyu and Fain, Brandon and Lyu, Liang and Munagala, Kamesh},
  booktitle={Proceedings of the 36th International Conference on Machine Learning (ICML-19)},
  pages={1032--1041},
  year={2019},
  organization={PMLR}
}

@article{mahabadi2024core,
  title={Core-sets for Fair and Diverse Data Summarization},
  author={Mahabadi, Sepideh and Trajanovski, Stojan},
  journal={Advances in Neural Information Processing Systems},
  volume={36},
  year={2024}
}

@inproceedings{wu2024new,
  title={New Algorithms for Distributed Fair k-Center Clustering: Almost Accurate as Sequential Algorithms},
  author={Wu, Xiaoliang and Feng, Qilong and Huang, Ziyun and Xu, Jinhui and Wang, Jianxin},
  booktitle={Proceedings of the 23rd International Conference on Autonomous Agents and Multiagent Systems (AAMAS-24)},
  pages={1938--1946},
  year={2024}
}

@inproceedings{backurs2019scalable,
  title={Scalable fair clustering},
  author={Backurs, Arturs and Indyk, Piotr and Onak, Krzysztof and Schieber, Baruch and Vakilian, Ali and Wagner, Tal},
  booktitle={Proceedings of the 36th International Conference on Machine Learning (ICML-19)},
  pages={405--413},
  year={2019},
  organization={PMLR}
}

@inproceedings{lin2024streaming,
  title={Streaming Fair k-Center Clustering over Massive Dataset with Performance Guarantee},
  author={Lin, Zeyu and Guo, Longkun and Jia, Chaoqi},
  booktitle={Proceedings of the 28th Pacific-Asia Conference on Knowledge Discovery and Data Mining (PAKDD-24)},
  pages={105--117},
  year={2024},
  organization={Springer}
}

@article{chen2024approximation,
  title={An approximation algorithm for diversity-aware fair k-supplier problem},
  author={Chen, Xianrun and Ji, Sai and Wu, Chenchen and Xu, Yicheng and Yang, Yang},
  journal={Theoretical Computer Science},
  volume={983},
  pages={114305},
  year={2024},
  publisher={Elsevier}
}

@ARTICLE{guo2025near,
  author={Guo, Longkun and Jia, Chaoqi and Liao, Kewen and Lu, Zhigang and Xue, Minhui},
  journal={IEEE Transactions on Neural Networks and Learning Systems}, 
  title={Near-Optimal Algorithms for Instance-Level Constrained $k$-Center Clustering}, 
  year={2025},
  volume={36},
  number={10},
  pages={18844-18858}}

@misc{wholesale_customers_292,
  author       = {Cardoso, Margarida},
  title        = {{Wholesale customers}},
  year         = {2013},
  howpublished = {UCI Machine Learning Repository},
  note         = {{DOI}: https://doi.org/10.24432/C5030X}
}

@misc{student_performance_320,
  author       = {Cortez, Paulo},
  title        = {{Student Performance}},
  year         = {2008},
  howpublished = {UCI Machine Learning Repository},
  note         = {{DOI}: https://doi.org/10.24432/C5TG7T}
}

@misc{bank_marketing_222,
  author       = {Moro, S. and Rita, P. and Cortez, P.},
  title        = {{Bank Marketing}},
  year         = {2014},
  howpublished = {UCI Machine Learning Repository},
  note         = {{DOI}: https://doi.org/10.24432/C5K306}
}

@misc{default_of_credit_card_clients_350,
  author       = {Yeh, I-Cheng},
  title        = {{Default of Credit Card Clients}},
  year         = {2009},
  howpublished = {UCI Machine Learning Repository},
  note         = {{DOI}: https://doi.org/10.24432/C55S3H}
}

@misc{adult_2,
  author       = {Becker, Barry and Kohavi, Ronny},
  title        = {{Adult}},
  year         = {1996},
  howpublished = {UCI Machine Learning Repository},
  note         = {{DOI}: https://doi.org/10.24432/C5XW20}
}

@inproceedings{liu2015faceattributes,
  title = {Deep Learning Face Attributes in the Wild},
  author = {Liu, Ziwei and Luo, Ping and Wang, Xiaogang and Tang, Xiaoou},
  booktitle = {Proceedings of the 15th International Conference on Computer Vision (ICCV-15)},
  month = {December},
  year = {2015} 
}

@inproceedings{ceccarello2024fast,
  title={Fast and accurate fair k-center clustering in doubling metrics},
  author={Ceccarello, Matteo and Pietracaprina, Andrea and Pucci, Geppino},
  booktitle={Proceedings of the 33rd ACM Web Conference (WWW-24)},
  pages={756--767},
  year={2024}
}

@inproceedings{0012YP24,
  author       = {He Zhang and
                  Xingliang Yuan and
                  Shirui Pan},
  title        = {Unraveling Privacy Risks of Individual Fairness in Graph Neural Networks},
  booktitle    = {Proceedings of the 40th IEEE International Conference on Data Engineering (ICDE-24)},
  pages        = {1712--1725},
  publisher    = {{IEEE}},
  year         = {2024}
}

@article{ZhangWYPTP24,
  author       = {He Zhang and
                  Bang Wu and
                  Xingliang Yuan and
                  Shirui Pan and
                  Hanghang Tong and
                  Jian Pei},
  title        = {Trustworthy Graph Neural Networks: Aspects, Methods, and Trends},
  journal      = {Proc. {IEEE}},
  volume       = {112},
  number       = {2},
  pages        = {97--139},
  year         = {2024}
}

@inproceedings{0012WYY0Y25,
  title={Dynamic Graph Unlearning: A General and Efficient Post-Processing Method via Gradient Transformation},
  author={Zhang, He and Wu, Bang and Yang, Xiangwen and Yuan, Xingliang and Liu, Xiaoning and Yi, Xun},
  booktitle={Proceedings of the ACM on Web Conference 2025},
  pages={931--944},
  year={2025}
}

\end{document}